\documentclass[draft,narroweqnarray,inline]{ieee}
\pdfoutput=1
\usepackage{graphicx}
\usepackage{psfrag}
\usepackage{amsmath,amssymb}
\usepackage[normalem]{ulem}

\def\nit{\hbox{\it I\hskip -2pt N}}

\def\classP{$\boldsymbol P$ }
\def\classNP{$\boldsymbol N \! \boldsymbol P$ }

\newtheorem{theorem}{Theorem}[section]
\newtheorem{definition}[theorem]{Definition}

\begin{document}


\title[The \classP - \classNP question and the pseudo-randomness of \classNP problems]{About the impossibility to prove \classP $\neq$ \classNP \sout{or \classP = \classNP} and the pseudo-randomness in \classNP}


\author[M.R\'emon]{Prof. Marcel R\'emon\authorinfo{%
M.R\'emon, Department of Mathematics,
      Namur University, Belgium;
        \mbox{marcel.remon@unamur.be}}
      }

\maketitle

\begin{abstract}
\noindent The relationship between the complexity classes \classP and
\classNP is an unsolved question in the field of theoretical computer science.
In this paper, we look at the
link between the \classP - \classNP question and the
{\it ``Deterministic''} versus {\it ``Non Deterministic''} nature of a
problem, and more specifically at the temporal nature of the complexity within the \classNP class of problems. Let us remind that the \classNP class is called the class of
{\it ``Non
Deterministic Polynomial''} languages.
Using the meta argument that
results in Mathematics should be {\it ``time independent''} as they are reproducible, the
paper shows that the \classP $\neq$ \classNP assertion is  impossible
to prove in the {\it a-temporal} framework of
Mathematics. {\bf In a previous version of the report, we use a similar
argument based on randomness to show that the \classP = \classNP
assertion was also impossible to prove, but this part of the paper was
shown to be incorrect.  So, this version deleted it.}   In fact, this paper highlights the time dependence of the complexity for any \classNP problem, linked to some pseudo-randomness in its heart.   

\end{abstract}




\section{Introduction}

\subsection{The class  \classP of languages}
\noindent A decision problem is a problem that takes as
input some string, and outputs "yes" or "no". If there is an algorithm
(say a Turing machine, or a computer program with unbounded memory)
which is able to produce the correct answer for any input string of
length $n$ in at most $c \; n^k$ steps, where $k$ and $c$ are constants independent of the input string, then we say that the problem can be solved in polynomial time and we place it in the class \classP$\!\!$.  \\[12pt]
\noindent More formally, \classP is defined as the set of all languages
which can be decided by a deterministic polynomial-time Turing
machine.  Here we follow the framework proposed by Stephen \cite{Cook2000}.  Let $\Sigma$ be a finite alphabet with at least two elements, and let $\Sigma^*$ be the set of finite strings over $\Sigma$.  Then a language over $\Sigma$ is a subset $L$ of $\Sigma^*$.  Each Turing Machine $M$ has an associated input alphabet $\Sigma$.  For each string $w$  in $\Sigma^*$, there is a computation associated with $M$, with input $w$.  We say that $M$ {\it accepts $w$} if this computation terminates in the accepting state ``{\it Yes}''. Note that $M$ fails to accept $w$ either if this computation ends in the rejecting state ``{\it No}'', or if the computation fails to terminate.  \\[12pt]
\noindent The {\it language accepted by} $M$, denoted $L(M)$, has associated alphabet $\Sigma$ and is defined by
\[ L(M) = \{ w \in \Sigma^* | M \mbox{ accepts } w \} \]

\noindent We denote by $t_M(w)$ the number of steps in the computation of $M$ on input $w$.  If this computation never halts, then $t_M(w) = \infty$.  For $n \in \nit$, we denote by $T_M(n)$ the worst case run time of $M$; that is 
\[T_M(n) = \max \{t_M(w) | w \in \Sigma^n\} \]
where $\Sigma^n$ is the set of all strings over $\Sigma$ of length $n$.  We say that $M$ {\it runs in polynomial time} if :
\[ \exists k \; \in \nit \; \mbox{ such that } \; \; \{ \forall n :  T_M(n) \leq n^k + k\;\;\} \]  
\begin{definition} We define the class ${\boldsymbol P}$ of languages by 
\begin{eqnarray*}
{\boldsymbol P} & = & \{ L | L = L(M) \mbox{ for a machine } M \mbox{ which runs in polynomial time} \} 
\end{eqnarray*}
\end{definition}
\subsection{The class  \classNP of languages}
\noindent The notation ${\boldsymbol N \! \boldsymbol P}$ stands for {\it non deterministic polynomial time}, since originally ${\boldsymbol N \! \boldsymbol P}$ was defined in terms of non deterministic machines.  However, it is customary to give an equivalent definition using the notion of a {\it checking relation}, which is simply a binary relation $R \subseteq \Sigma^* \times \Sigma_1^*$ for some finite alphabets $\Sigma$ and $\Sigma_1$.  We associate with each such relation $R$ a language $L_R$ over 
$\Sigma \cup \Sigma_1 \cup \{\#\}$ defined by
\[ L_R = \{ w \# y | R(w,y) \} \]
where the symbol $\#$ is not in $\Sigma$.  We say that $R$ is {\it polynomial-time} iff $L_R \in \boldsymbol P$. 

\begin{definition} We define the class $\boldsymbol N \! \boldsymbol P$ of languages by the condition that a language $L$ over $\Sigma$ is in 
$\boldsymbol N \! \boldsymbol P$ iff there is $k \in \nit$ and a polynomial-time checking relation $R$ such that for all $w \in \Sigma^*$,
\[ w \in L \Leftrightarrow \exists y ( |y| \leq |w|^k \mbox{ and } R(w,y) ) \]
where $|w|$ and $|y|$ denote the lengths of $w$ and $y$, respectively.  We say that $y$ is {\it a certificate} associated to $w$.
\end{definition}
\subsection{The \classP - \classNP question}
\noindent The {\it ``\classP versus \classNP problem''}, i.e. the
question whether \classP = \classNP or \classP $\neq$ \classNP, is an open question and is
the core of this paper.  See \cite{Sipser92}
 for the history of the question.  Here, we show that neither \classP = \classNP nor \classP $\neq$ \classNP can be proved in the {\it ``a-temporal''} framework of Mathematics where results should always be reproducible.  We link this assertion to the existence of some pseudo-random part in the heart of any \classNP problem.
\subsection{An example of \classNP problem : the 3-CNF-satisfiability problem}
\noindent {\it Boolean formulae} are built in the usual way from
propositional variables $x_i$ and the logical connectives $\wedge$, $\vee$ and $\neg$, which are interpreted as conjunction, disjunction, and negation, respectively.  A {\it literal} is a propositional variable or the negation of a propositional variable, and a {\it clause} is a disjunction of literals.  A Boolean formula is {\it in conjunctive normal form} iff it is a conjunction of clauses. \\[12pt]
\noindent A {\it 3-CNF formula} $\varphi$ is a Boolean formula in conjunctive normal form with exactly three literals per clause, like $\varphi := (x_1 \vee x_2 \vee \neg x_3) \wedge (\neg x_2 \vee  x_3 \vee \neg x_4) := \psi_1 \wedge \psi_2 $.  The {\it 3-CNF-satisfiability or 3-CNF-SAT problem} is to decide whether there exists or not logical values for the literals so that $\varphi$ can be true (on the previous example, $\varphi = 1 $(True) if  $x_1 = \neg x_2  = 1)$.   \\[12pt]
\noindent Until now, nobody knows whether or not it is possible to check the satisfiability of any given {\it 3-CNF} formula $\varphi$ in a polynomial time, as the {\it 3-CNF-SAT} problem is known to belong to the class \classNP of problems. See \cite{cormen2001} for details.  \\[12pt]
\noindent Let us give some general properties of the {\it 3-CNF} formulae. \\[12pt]
\noindent The {\it size} $s$ of a 3-CNF formula $\varphi$ is defined
as the size of the corresponding {\it Boolean circuit}, i.e. the
number of logical connectives in  $\varphi$.  Let us note
the following property of the size $s$ :
\begin{eqnarray}
 s = {\cal O}(m) = {\cal O}(n^3)
\label{no0}
\end{eqnarray}
\noindent  where
 $n$ is the number of propositional
variables $x_i$ and $m$ the number of clauses in $\varphi$.  Indeed,  
\[ \frac{n}{3} \leq m \leq 2^3 \frac{n (n-1) (n-2)}{3
  \times 2} \mbox{\hspace{0,5cm} and \hspace{0,5cm}} (3 m - 1) \leq s
\leq (6 m - 1) \]
%
%
\noindent as there is a maximum of $2^3 \times C_3^n$
possible clauses which corresponds to the choice of 3 different
variables among $n$, each of them being in an affirmative or negative
state. Note that $s = 3 m -1$ when there is no ``$\neg$'' in $\varphi$ [$m \times 2$ logical connectives ``$\vee$'' for the $\psi_i$ and $m-1$ ``$\wedge$'' as conjonctions] and 
$s = 6 m -1$ when all the litterals in $\varphi$ are in a negative form.\\[12pt]
\noindent In this paper, we define the {\it dimension} $d$ of a 3-CNF
formula as $(n, m)$.  And we represent any 3-CNF formula by a matrix ${\cal
  A}$ of size $2 n
\times m$.  The {\it signature $u_i$ of a clause $\psi_i$} is defined as the value
of the
binary number corresponding to the row in the matrix.  The {\it
  signature of a formula} is the ordered vector of these clause's signatures 
: $\varphi_{n,m} \approx (u_1,u_2,\cdots,u_m)$ with 
$21 \leq u_i \leq 21 \cdot 2^{2n-5}$ and $u_i > u_j$ for $i <j$. See {\tt Table I.} 
\\[6pt]
%
\noindent 
\begin{table}[htb]
\begin{center}
{\small
\begin{tabular}{rcl|cc|cc|cc|cc|c|c|}
\cline{4-11} 
& & & \multicolumn{8}{|c|}{3-CNF formula $\varphi$ ({\it dimension} $d
  = (4, 3) $)} &  \multicolumn{2}{c}{} \\
\cline{4-11} \cline{13-13}
& & & $x_1$ & $\neg x_1$ & $x_2$ & $\neg x_2$ & $x_3$
& $\neg x_3$ & $x_4$ & $\neg x_4$ & &  $u_i$
\\
\cline{4-11} \cline{13-13}
$\psi_1 $ :& $( x_1 \vee x_2 \vee \neg x_3 )$ & & 1 & 0 & 1 & 0 & 0 &
1 & 0 & 0 & & \parbox[c][14pt]{12pt}{$164$ } \\
$\wedge \; \; \psi_2$ : & $( \neg x_2 \vee x_3 \vee \neg x_4) $ &
$\Leftrightarrow$  & 0
& 0 & 0 & 1 & 1 & 0 & 0 & 1 & & \parbox[c][14pt]{12pt}{$25$}  \\
$\wedge \; \; \psi_3$ :  & $( \neg x_1 \vee \neg x_3 \vee x_4) $ & & 0 & 1 & 0 & 0 & 0 & 1
& 1 & 0 & & \parbox[c][14pt]{12pt}{$70$} \\
\cline{4-11} \cline{13-13}
\end{tabular}
}
\end{center}
\mbox{}\\[6pt]
\caption{Example of matrix representation and signatures of a 3-CNF formula.}
\end{table}
\mbox{}\vspace{15pt}\\
\newpage
\noindent There are  $2^3 \times C_3^n$
possible clauses with $n$ variables.  A 3-CNF formula with {\it dimension} $(k,m)$ with $k \leq n$ is composed of $m$ different clauses drawn from
the $2^3 \times C_3^n$ possible clauses.  So, the total number of such
formulae is 
\begin{eqnarray}
C_m^{2^3 \times C_3^n} = \frac{ (2^3 \times C_3^n)!}{m! \times (2^3
  \times C_3^n - m)!} = {\cal O}(n^{3m})
\label{no1}
\end{eqnarray}
\noindent Let $\Phi_{n,m}$ denote the set of all these formulae :  
\[ \Phi_{n,m} = \{\varphi : \varphi \mbox{ is a 3-CNF
    formula of {\it dimension} $(k, m) $ with $ k \leq n$ \}   } \]
%
%

\noindent 
The 3-CNF-Satisfiability
problem is to find  a function $\Xi$ : 
\begin{eqnarray}
 \Xi : & \Phi_{n,m} & 
\longrightarrow  \{0,1\} \label{no3} \\
& \varphi & \; \leadsto  \mbox{\hspace{0.4cm} $0$ if $\varphi$ is non satisfiable and
  $1$ otherwise }  \nonumber
\end{eqnarray}

\noindent The 3-CNF-Satisfiability problem is known to belong to the
\classNP class. 
\mbox{}\\[3pt]
\section{A ``Meta Mathematical'' proof that \classP $\neq$ \classNP is impossible to prove} 
\mbox{}\\[3pt]
\noindent One way to prove that \classP $\neq$ \classNP  is to show that the complexity measure  $T_{M}(n)$ for some \classNP problem, like the 3-CNF-SAT problem, cannot be reduced to a polynomial time.  We will show that the 3-CNF-SAT problem behaves as a common {\it safe problem} and that its complexity is time dependent.  In fact, at some specific time  $t_0+\Delta t$, the 3-CNF-SAT problem will be of polynomial complexity. So, \classP $\neq$ \classNP will not be provable, as $T_{M}(n)$ is not {\it ``always''} supra-polynomial.
\subsection{The analogy with the {\it safe problem} and the time dependent nature of complexity} 
%
%
\noindent Finding whether or not a given 3-CNF formula $\varphi$ is satisfiable is like being in front of a {\it safe}, trying to find the {\it opening combination}.  One has to try any possible value (0 or 1) for the variable $x_i$ in $\varphi$ to see whether some combination satisfies $\varphi$, in the same way as one tries any combination to get the one, if it exists, that opens the safe. \\[12pt]
\noindent Let us consider more deeply the analogy between the 3-CNF-SAT problem and the {\it safe problem}, especially by looking to the {\it time dependent} nature of the complexity involved here.  It is clear that when you are in front of a safe for the first time, it is a very hard problem, as you do not have any information about the correct opening combination.  In fact, in the worst case, it takes an exponential time to find it.  But as soon as  you have succeeded in opening the safe (or in finding that there is no solution), the problem becomes trivial.  It takes only one operation to open the safe or to declare it impossible to open.\\[12pt]
\noindent Let us denote by $t_0$ the first time you try to open the safe, and by $\Delta
t$ the time needed to find the solution.  Let us remark that $\Delta t$ can be
huge but it is always finite as the number of possible combinations is finite.  Now we compute the complexity measure $T_{safe}(n)$ for
 the {\it safe problem} at $t_0$ and $t_0 + \Delta t$.\\[12pt]
\noindent In $t_0$, one has to test all possible combinations.  If the safe has $n$ buttons with only two positions (0 or 1), there will be $2^n$ possibilities.  Because no information is available about the solution, there is no way to reduce the number of cases to be tested.  The exponential complexity of the problem comes from the total lack of information about the solution.  This absence of information is strictly related to {\it the random nature of the problem} : the finding of the opening combination is a random search process for anyone in front of the safe, at least in $t_0$.  So, we get 
\[ T_{safe, \; t_0}(n) = 2^n \]
\noindent But after $\Delta t$, the correct opening combination is known {\it forever}, and the complexity measure is now 
\[ T_{safe, \; t_0 + \Delta t}(n) = 1 \]
\noindent As one can see, the complexity measure $T_{safe}(n)$ for the safe problem is {\it time dependent}.\\[12pt]
\noindent The same occurs for the 3-CNF-SAT problem as well as for any \classNP problem.
Their complexity measure changes in time.  The idea of this section about the
impossibility to prove \classP $\neq$ \classNP is to show that, even if
$T_{3-CNF-SAT, \; t_0}(n)$ is not known (exponential or polynomial ?), there
exists some $\Delta t$, even huge, such that the complexity measure is polynomial in $t_0 + \Delta t$.  
\subsection{The Computation of $T_{3-CNF-SAT, \; t_0+\Delta t}(n)$}

\noindent Let us take $\Delta t$ large enough so that $\Xi$ [the 3-CNF-SAT
decision function, see equation (\ref{no3})] is known for all the
3-CNF formulae in $\Phi_{n,m}$.  $\Delta t$ exists and is finite.  In the analogy with the safe problem, it corresponds to the time needed to find the solution for all safe equipments of dimension $n$.
 Until now, we do not know whether $\Xi$ can be computed  in
polynomial time or not, but this only changes the size of $\Delta t$.\\[12pt]
\noindent The output of $\Xi$ is the set ${\cal S}_{n,m}$ of all satisfiable 3-CNF
formulae of $\Phi_{n,m}$, or equivalently $\overline{\cal S}_{n,m} = \Phi_{n,m}
\setminus {\cal S}_{n,m}$, the set of all non satisfiable 3-CNF formulae. As equation
(\ref{no1}) shows, $\overline{\cal S}_{n,m}$
contains at most ${\cal O}(n^{3m})$ 
elements.  The worst case occurs when $m = (2^3 \times C_3^n)/2 =
{\cal O}(n^3)$.
 As $\overline{\cal S}_{n,m} \subseteq \Phi_{n,m}$, the
equation (\ref{no1}) gives us the following result : 
\begin{eqnarray}
\#\{\overline{\cal S}_{n,m}\} < \#\{\Phi_{n,m}\} = {\cal O}(n^{3(n^3)})
\hspace{0.3cm} \Rightarrow  \hspace{0.3cm} \#\{\overline{\cal S}_{n,m}\} = {\cal O}(2^{n^3}) \hspace{0.3cm} \mbox{ as $n^3 > 2$}
\label{no4bis}
\end{eqnarray} 
\noindent See Figure \ref{figure_4_m} for an example of $\# \{\Phi_{n,m}\}$ and $\# \{\overline{\cal S}_{n,m} \}$ with $n = 4$.  The figure shows that $\# \{\Phi_{n,m}\}$ and $\# \{\overline{\cal S}_{n,m} \}$ behaves similarly.\\[12pt] 
\begin{figure}[htb]
\centerline{\includegraphics[width=12cm]{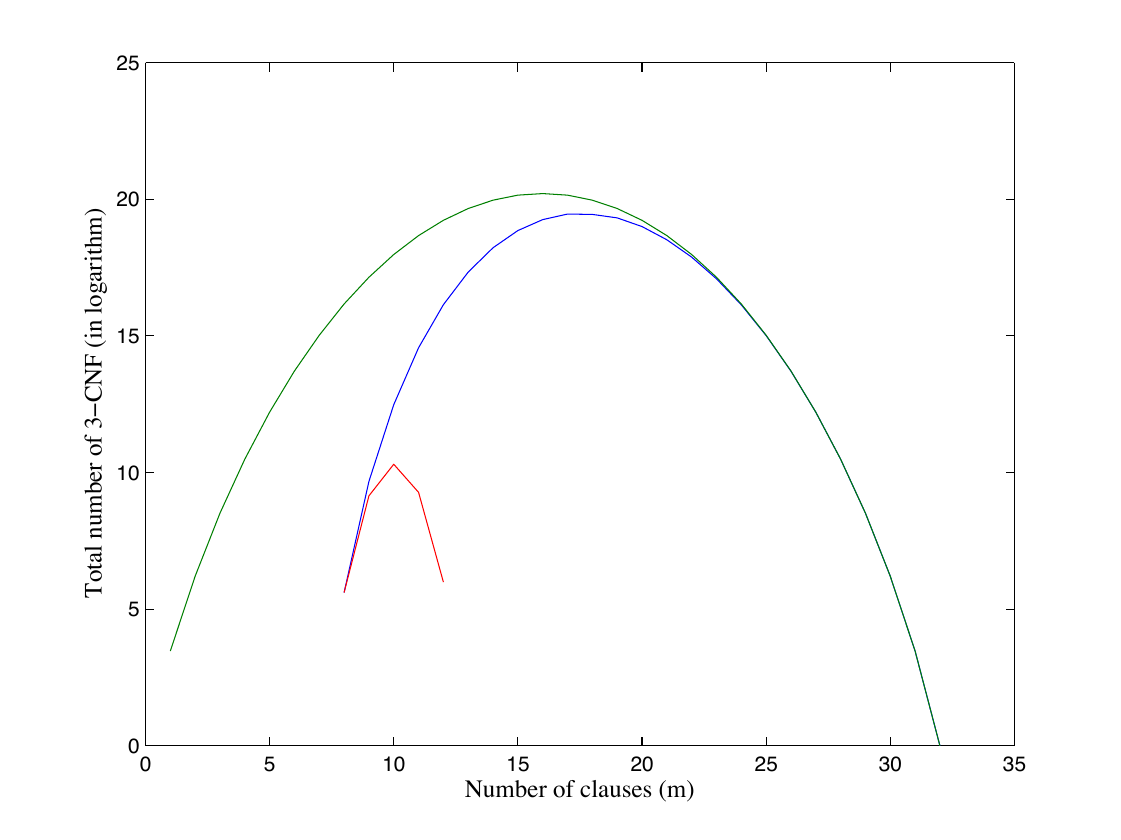}}
\caption{Logarithmic scale : the upper curve represents the total
  number of all possible 3-CNF in $\Phi_{4,m}$;
 the second one, the total of
  non-satisfiable 3-CNF, {\it i.e.} $\# \{\overline{\cal S}_{4,m} \}$,
and the lower one, the total of
{\it Irreducible Non-Satisfiable} 3-CNF, {\it i.e.} $\# \{\overline{\cal S}_{4,m}^{INS} \}$ (i.e. 3-CNF Satisfiable with m-1 clauses).}
\label{figure_4_m}
\end{figure}  
%
%
\noindent So, one can now calculate $T_{3-CNF-SAT, \; t_0+\Delta t}(n)$ : it is the time
required to check whether a specific 3-CNF formula belongs or not in
$\overline{\cal S}_{n,m}$, after $\Delta t$ large enough for the entire set $\overline{\cal S}_{n,m}$ to be computed. If
one can allocate an exponential space for memory to save the
elements of $\overline{\cal S}_{n,m}$ (as accepted in Turing machines), then a hash algorithm,
based on the {\it clause's signatures}, can be used to see whether a 3-CNF formula $\varphi$ belongs or not
to the set $\overline{\cal S}_{n,m}$.  For instance, one can use $u_i$, the {\it $i^{th}$ ordered signature of clauses}, as the $i^{th}$ successive hash function $h_i(\varphi)$. It
takes ${\cal O}(2 n)$ operations to compute each
of these $m$
{\it clause's signatures} of $\varphi$ and ${\cal O}(m \log m)$ computations to
sort them.  We need then ${\cal O}(2^3 \times C_3^n)$ operations,
which corresponds to the maximum number of possible values for the signatures, to find whether the signature
belongs or not to the corresponding section of $\overline{\cal S}_{n,m}$ where the formulae are also ordered, in a lexical ordering, following their clause's signatures.  Using
equation (\ref{no0}) [{\it i.e.} ${\cal O}(m) = {\cal O}(n^3)$], 
\begin{eqnarray}
T_{3-CNF-SAT, \;
  t_0+\Delta t}(n)&=&{\cal O}(m(2n) + (m \log m) + m(2^3
C_3^n)) \nonumber \\
&=&{\cal O}(m^2) = {\cal O}(n^k)  \;\;\; \mbox{ for some }
k \in \nit
\label{no4}
\end{eqnarray}  
%
\subsection{The ``unprovability'' of  \classP $\!\!
\neq \!\!$ \classNP}
\begin{theorem}  It is impossible to prove that \classP $\!\!
\neq \!\!$ \classNP in the deterministic or time independent framework of Mathematics. 
\end{theorem}
\begin{proof}
The solution of the 3-CNF-SAT problem
is equivalent to the setting of these two functions $\Xi'$ and $\Xi"$ : 
\begin{eqnarray}
\mbox{(In $t_0$) \hspace{0.3cm}} & \Xi'  : & \Phi_{n,m}  
\stackrel{{\cal O}(?)}{\longrightarrow} \{0,1\} \mbox{\it \hspace{1cm}
  (the construction of $\overline{\cal S}_{n,m}$)} \nonumber \\ 
& & \hspace{0.3cm} \varphi \hspace{0.4cm} \leadsto \hspace{0.4cm} 0
\mbox{\hspace{0.4cm}  if } \varphi \in \overline{\cal S}_{n,m} \mbox{ and $1$
  otherwise }
\label{no5} \\
\mbox{(In $t_0 + \Delta t$) \hspace{0.3cm}} & \Xi''  : & \Phi_{n,m}  
\stackrel{{\cal O}(n^k)}{\longrightarrow} \{0,1\} \mbox{\it \hspace{1cm}
  ($\varphi \stackrel{?}{\in} \overline{\cal S}_{n,m}$ when
  $\overline{\cal S}_{n,m}$ is known)} \nonumber \\
& & \hspace{0.3cm} \varphi \hspace{0.4cm} \leadsto \hspace{0.4cm} 0
\mbox{\hspace{0.4cm}  if } \varphi \in \overline{\cal S}_{n,m} \mbox{ and $1$
  otherwise }
\label{no6} 
\end{eqnarray}
\noindent The {\it meta mathematical argument}
lies in the fact that any operation done by
$\Xi'$ in $t_0$ can be reduced to a polynomial time operation by $\Xi''$ in $t_0 +
\Delta t$ \footnote{To make it easier to understand, let us think of
the version of 3-CNF-SAT with $n = 4$ : it took us several months
to build  $\overline{\cal S}_{n,m}$, but now it only takes
seconds to solve the 3-CNF-SAT problem with $4$ variables.  And this
is done forever.  A similar reasoning can be done for the
$i^{th}$ decimal of $\pi$, or for the
list of the $n$ first prime numbers.}. \\[12pt]
\noindent Mathematically speaking, it is impossible to make a formal or mathematical
distinction between
both functions $\Xi'$ and $\Xi"$, as time does not interfere with proofs in
mathematics.
More precisely, if someone proves that the 3-CNF-SAT problem $\Xi$ (or $\Xi'$) is non polynomial,
this assertion, as well as the steps for the demonstration, should be true at any time, independently of $t$, even in $t_0
+ \Delta t$.  The proof could not introduce time in the demonstration.
But people will only be able to proof the non polynomial nature of
3-CNF-SAT for time $t_0$, certainly not for time $t_0 + \Delta t$ as
shown in equation (\ref{no4}).  And this argument holds for all \classNP problems because all of them are equivalent, in term of complexity, to the 3-CNF-SAT problem.
\end{proof}
This is exactly the same situation as with the safe problem : the complexity
measure of the problem is changing over time, becoming polynomial after some
large $\Delta t$.  But the \classP - \classNP question does not consider time as far as complexity is concerned : if we do not consider the time dependent nature of complexity, one should conclude that  \classP = \classNP.  
\mbox{}\\[-3pt]
\section{Conclusions}
\mbox{}\\[6pt]
\noindent This paper tries to show that the \classP $\!\!
\neq \!\!$ \classNP 
problem is impossible to solve within the time independent framework 
of Mathematics, as 
\classP $\!\! \neq \!\!$ \classNP can be proved without reference to time.
 The key concept of the paper is the temporal nature of the complexity measure
 for the \classNP$\!\!-$hard problems.  This time dependence is closely
 related to some (pseudo) randomness in the heart of these problems. Some
 analogy can be found with the Chaos theory, when pseudo randomness arises
 from deterministic processes. \\[12pt]
\noindent For the author, \classNP  is really different from \classP but the difference lies in the distinction between true randomness and mathematical pseudo-randomness, and this frontier is situated on the limit border of Mathematics (which is deterministic). \\[12pt]
\noindent The impossibility to prove that \classP $\neq$ \classNP gives a new perspective on the
{\it pseudo non deterministic (or random)} nature of the most difficult 
problems,
the \classNP$\!\!-$hard problems~: we can see these problems as so
inextricable that we are in front of them like someone facing some random
search problem (as the safe problem), even if they are deterministic (not
random) in their very essential nature, {\it i.e.} as quasi chaotic problems.
 \\[12pt]
Therefore, the   \classP $\!\! \neq
\!\!$ \classNP {\it ``unprovability''} can be seen as the expression of the
 incapa\-city for Mathematics to give a time independent definition of
 randomness. \\[12pt]

\nocite{Sipser92}
\nocite{Sanjeev09}

\addcontentsline{toc}{part}{\mbox{Bibliography}}  
\bibliographystyle{plain}   
\bibliography{mybib}

\end{document}